\documentclass[notitlepage,jmp]{revtex4-1}
\usepackage{graphicx}
\usepackage{amsmath,amssymb,amsthm}
\usepackage{lineno,hyperref}

\begin{document}

\title{Torus Solutions to the Weierstrass-Enneper Representation of Surfaces}
\author{Christopher Levi Duston}
\email{dustonc@merrimack.edu}
\affiliation{Physics Department, Merrimack College, N Andover MA, 01845}


\newtheorem{theorem}{Theorem}
\newtheorem{definition}{Definition}
\newtheorem{lemma}{Lemma}
\newtheorem{prop}{Proposition}

\date{\today}

\begin{abstract}
In this paper we present a torus solution to the generalized Weierstrass-Enneper representation of surfaces in $\mathbb{R}^4$. The key analytical technique will be Bloch wave functions with complex wave vectors. We will also discuss some possible uses of these solutions which derive from their explicit nature, such as Dehn surgery and the physics of exotic smooth structure.
\end{abstract}

\maketitle


\section{Background}
The Weierstrass-Enneper representation of a constant mean curvature (CMC) surface is a classical approach to explicitly specify the coordinates of such a surface in $\mathbb{R}^3$ \cite{Oprea-2007}. For a holomorphic function $f$ and a meromorphic function $g$ such that $fg^2$ is holomorphic, the coordinates of a minimal surface (that is, zero mean curvature) are given by
\[x(z,\bar{z})=\left(\text{Re}\int f(1-g^2)dz',\text{Re}\int if(1+g^2)dz',\text{Re}2\int fgdz'\right).\]
This representation has been extended to surfaces of any constant mean curvature using a linear system which formally resembles a Dirac equation \cite{Konopelchenko-1996},
\[\partial_z \psi=p\phi,\qquad \partial_{\bar{z}}\phi=-p\psi,\]
where the coordinates are specified by
\[x_1+ix_2=i\int_{\Gamma}(\bar{\psi}^2dz'-\bar{\phi}^2d\bar{z}'),\quad x_1-ix_2=i\int_{\Gamma}(\phi^2dz'-\phi^2d\bar{z}'),\quad x_3=-\int_\Gamma (\bar{\phi}\psi dz'+\psi\bar{\phi}d\bar{z}').\]
Here $\psi$ and $\phi$ are complex functions and $p(z,\bar{z})$ is a real-valued function. The coordinates do not depend on the choice of curve $\Gamma\in \mathbb{C}$, but only on the end points. It should also be emphasized that any constant mean curvature surface in $\mathbb{R}^3$ can be represented in this manner.

There are a number of further generalizations of this system, which we will collectively refer to as the Konopelchenko-Landolfi (KL) system, and will be the starting point for our analysis. By considering a set of complex functions that satisfy the Dirac-like equation
\begin{equation}\label{eq:Dirac}
  \partial_z\phi_\alpha=p\phi_\alpha,\qquad \partial_{\bar{z}}\phi_\alpha=-p\psi_{\alpha},\qquad \alpha=1,2,
  \end{equation}
a conformally immersed surface in $\mathbb{R}^4$ is specified by the following coordinates \cite{Konopelchenko-Landolfi-1999,Chen-Chen-2007}:
\begin{align}\label{eq:Weier}
\begin{split}
  x_1&=\frac{i}{2}\int_{\Gamma}\left\{\left(\bar{\psi}_1\bar{\psi}_2+\phi_1\phi_2\right)dz'-\left(\psi_1\psi_2+\bar{\phi}_1\bar{\phi}_2\right)d\bar{z}'\right\},\\
  x_2&=\frac{1}{2}\int_{\Gamma}\left\{\left(\bar{\psi}_1\bar{\psi}_2-\phi_1\phi_2\right)dz'+\left(\psi_1\psi_2-\bar{\phi}_1\bar{\phi}_2\right)d\bar{z}'\right\},\\
   x_3&=-\frac{1}{2}\int_{\Gamma}\left\{\left(\bar{\psi}_1\psi_2+\bar{\psi}_2\phi_1\right)dz'+\left(\psi_1\bar{\phi}_2+\phi_2\bar{\phi}_1\right)d\bar{z}'\right\},\\
     x_4&=-\frac{i}{2}\int_{\Gamma}\left\{\left(\bar{\psi}_1\psi_2-\bar{\psi}_2\phi_1\right)dz'-\left(\psi_1\bar{\phi}_2-\phi_2\bar{\phi}_1\right)d\bar{z}'\right\}.
\end{split}
\end{align}
  The induced metric on the surface is given by
  \[ds^2=u_1u_2dzd\bar{z},\]
  where $u_\alpha=|\psi_\alpha|^2+|\phi_\alpha|^2$. This has been the starting point for a number of studies; a partial list includes the integrable deformation of surfaces \cite{Taimanov-1997a,Konopelchenko-Landolfi-1999}, sigma-model dynamics \cite{Bracken-Grundland-Martina-1999,Bracken-Grundland-2002}, foliations \cite{Friedrich-1998}, cosmic strings \cite{Duston-2013}, and exotic smooth structure \cite{Asselmeyer-Maluga-Brans-2011, Asselmeyer-Maluga-Brans-2015}. The properties of the representations and particular solutions have also been explored \cite{Bracken-Grundland-1999, Bracken-Grundland-2001,Bracken-Grundland-2002,Bracken-2018}.

  Most known solutions to the KL equations are explicit solutions to the Dirac equation (\ref{eq:Dirac}), which means the resulting surface is analytically specified but it's precise geometric properties are unknown. When the KL equations (\ref{eq:Weier}) can be integrated explicitly, the surface can at least be parametrized. For example, by adding additional differential constraints\cite{Bracken-Grundland-2002}, one can find a rational solution of the form
  \[\frac{\psi_\alpha}{\bar{\phi}_\alpha}=e^z,\]
  which results in a CMC-surface parametrized by
  \[x_1^2+x_2^2=t(t-1)^2(3+3t+t^2)^2,\qquad x_3^2+x_4^2=4(1+t)^6,\qquad t=e^{z+\bar{z}}.\]
  However, there are no explicit torus solutions to the KL system. The algebraic parametrization of such a solution should be far simpler then the example above, but the KL representation of a torus is not immediately obvious. In this paper we will present such a solution, which will be a new solution to the KL system, as well as augment the list of applications given above. Furthermore, tori in $\mathbb{R}^4$ play a special role in differential geometry since they can be used to construct arbitrary 3-manifolds via Dehn surgery. We will briefly discuss this interesting application after the presentation of the main results.
  
\section{Bloch Wave}
To generate our torus solution, we will use Bloch waves for inspiration. These are solutions to wave functions in periodic lattice structures taking the form
\[\phi(x)=e^{-ikx}u(x),\]
for the arbitrary periodic function $u(x)$. We expect our solutions to be doubly-periodic in nature, and in fact the specific conditions on the Dirac equation (\ref{eq:Dirac}) have been previously specified \cite{Taimanov-1997}:
\[p(z+\gamma)=p(z)\qquad \phi_j(z+\gamma)=\epsilon(\gamma)\phi_j(z),\qquad \epsilon(\gamma)=\pm 1.\]
Here $\gamma\in\Lambda$, a lattice of rank 2. As ansatz we will choose solutions to (\ref{eq:Dirac}) as
\begin{align}\label{eq:Ansatz}
  \begin{split}
  \phi_\alpha&=A_\alpha\exp \left(i(k_\alpha z+h_\alpha \bar{z}\right)\\
  \psi_\alpha&=B_\alpha \exp \left( i(k_\alpha z+ h_\alpha \bar{z} \right).
\end{split}
  \end{align}
Here the constants $A_\alpha$, $B_\alpha$ and wave vectors $k_\alpha$, $h_\alpha$ are unspecified, but we will need to choose them in $\mathbb{R}$ or $\mathbb{C}$ as required. The double-periodicity will be enforced by choosing the wave vectors to satisfy conditions of the form
\[k_\alpha \gamma-h_\alpha \gamma=n\pi.\]

\section{Torus Solutions}
First we will require our ansatz (\ref{eq:Ansatz}) to satisfy the Dirac equation (\ref{eq:Dirac}). Doing that we arrive at the conditions
\[|p|^2=k_1h_1=k_2h_2.\]
Since the potential $p$ shows up in the metric, the resulting geometry will depend on the choice of wave vectors, which in turn depend on the underlying lattice structure. This underlying lattice structure will be described with a complex number 
\[\gamma=\Lambda_1+i\Lambda_2\]
defining the fundamental period.

\begin{theorem}
  The set of equations
  \begin{align}\label{eq:WeierProof}
    \begin{split}
      \phi_1&=C\exp\left(i(k_1z+h_1\bar{z}))\right),\\
      \phi_2&=\frac{2(k_1+k_2)}{C}\exp\left(i(k_1z+h_2\bar{z})\right),\\
      \psi_1&=-C\exp\left(i(k_1z+h_1\bar{z})\right),\\
      \psi_2&=\frac{2(h_1+h_2)}{C}\exp\left(i(k_2z+h_2\bar{z})\right),
    \end{split}
  \end{align}
  with constant $C\in\mathbb{C}$ and evanescent wave vectors
  \begin{align*}
    k_1&=a+ib\\
    h_1&=\left(\frac{n\pi}{\Lambda_1}-a\right)+i\left\{\left(\frac{n\pi}{\Lambda_1}-2a\right)\frac{\Lambda_2}{\Lambda_1}-b\right\}\\
    k_2&=\frac{n\pi}{2\Lambda_1}+i\left(\frac{n\pi}{2\Lambda_1}-a\right)\frac{\Lambda_2}{\Lambda_2}\\
    h_2&=\left(a-\frac{n\pi}{2\Lambda_1}\right)+i\left(a+\frac{n\pi}{2\Lambda_1}\right)\frac{\Lambda_2}{\Lambda_1}.
  \end{align*}
  with $a,b\in\mathbb{R}$ and $n\in\mathbb{Z}$ describe a torus under the Weierstrass representation (\ref{eq:Weier}). 
  \end{theorem}
\begin{proof}
  First assume the wave vectors in the ansatz (\ref{eq:Ansatz}) could be complex (evanescent), and we will prove the above represents a torus by directly finding the coordinates (\ref{eq:Weier}). Defining
  \[\eta=(k_1+k_2)z+(h_1+h_2)\bar{z},\]
  the first coordinate takes the following form:
  \[x_1=\frac{i}{2}\int_\Gamma \left\{A_1A_2\exp(i\eta)dz'-\bar{A}_1\bar{A}_2\exp(-i\bar{\eta})d\bar{z}'-B_1B_2\exp(i\eta)d\bar{z}'+\bar{B}_1\bar{B}_2\exp(-i\eta)dz\right\}.\]
  To facilitate this integration, define a new variable $u=\exp(i\eta)$, and by setting the constants with
  \begin{equation}\label{eq:AB1}
    A_1A_2=2(k_1+k_2),\qquad B_1B_2=-2(h_1+h_2),
    \end{equation}
  the first coordinate simply becomes
  \[x_1=\int_{\Gamma}(du+d\bar{u})=2 \Re\int_\Gamma du=2\Re(u),\]
  where we have chosen an arbitrary integration path from the origin to $z$. The calculation of the second coordinate proceeds similarly,
  \[x_2=i(u-\bar{u})=-2\Im(u).\]

  For the second pair of coordinates, we want to define a variable
  \[\rho=(k_1-\bar{h}_2)z+(h_1-\bar{k}_2)\bar{z},\]
  so that
  \[x_3=-\frac{1}{2}\int_\Gamma \left\{\bar{B}_1A_2\exp(-i\bar{\rho})dz'+\bar{B}_2A_1\exp(i\rho)dz+B_1\bar{A}_2\exp(i\rho)d\bar{z}'+B_2\bar{A}_1\exp(-i\rho)d\bar{z}'\right\}.\]
  Then by setting the conditions
  \begin{equation}\label{eq:AB2}
    -\frac{1}{2}\bar{B}_1A_2=(\bar{h}_1-k_2),\qquad -\frac{1}{2}\bar{B}_2A_1=-(k_1-\bar{h}_2),
    \end{equation}
  and using $w=\exp(i\rho)$, we have
  \[x_3=i\int_\Gamma dw-d\bar{w}=-2\Im(w)\]
  \[x_4=\int_{\Gamma}dw+d\bar{w}=2\Re(w).\]

  To get the precise form of the functions (\ref{eq:WeierProof}), we use the relationships (\ref{eq:AB1}) and (\ref{eq:AB2}), along with the following consistency conditions
  \begin{equation}\label{eq:ConsCond}
    \bar{k}_1-h_2=h_1+h_2,\qquad h_1-\bar{k}_2=\bar{k}_1+\bar{k}_2
    \end{equation}
  and by setting $A_1=C$.
  
  Defining a real parameter $t=z+\bar{z}$, these coordinates can be rewritten as
  \begin{align*}
      x_1&=2\cos((k_1+k_2)t),\\
      x_2&=-2\sin((k_1+k_2)t),\\
      x_3&=-2\sin((h_1+h_2)t),\\
      x_4&=2\cos((h_1+h_2)t).
    \end{align*}
  It is now easy to see that these coordinates define a flat torus, since
  \[x_1^2+x_2^2=4,\qquad x_3^2+x_4^2=4.\]

  Now to determine conditions placed on the complex wave vectors, we consider the periodic condition, first just on $\phi_1$ (which will have the same structure as $\psi_1$):
  \[\phi_1(z+\gamma)=\exp(i(k_1 z+h_1\bar{z}))\exp(ik_1(\Lambda_1+i\Lambda_2)+ih_1(\Lambda_1-i\Lambda_2)).\]
  The exponent in the second term must be equal to $in\pi$ for $n\in\mathbb{Z}$.
  By explicitly breaking the wave vectors into real and imaginary parts,
  \[k_1=a+ib,\qquad h_1=c+id,\]
  the parameters must satisfy
  \[(a+c)\Lambda_1=n\pi,\qquad (b+d)\Lambda_1=(c-a)\Lambda_2.\]
  Therefore these two wave vectors only depend on two real parameters $a,b\in\mathbb{R}$:
  \[k_1=a+ib\qquad h_1=\left(\frac{n\pi}{\Lambda_1}-a\right)+i\left\{\left(\frac{n\pi}{\Lambda_1}-2a\right)\frac{\Lambda_2}{\Lambda_1}-b\right\}\]
  Following the same approach for $\psi_2$ (and therefore $\phi_2$), we start with
  \[k_2=a_2+ib_2,\qquad h_2=c_2+id_2,\]
  use the consistency conditions (\ref{eq:ConsCond}), we determine the relationship of these $a_2,b_2,c_2,d_2\in\mathbb{R}$ with $a$ and $b$, as shown in the statement of the theorem.
  
  \end{proof}

We also note that if we are going to require that the metric components be real,
\[p^2=k_1h_1=a\left(\frac{n\pi}{\Lambda_1}-a\right)-b\left\{\left(\frac{n\pi}{\Lambda_1}-2a\right)\frac{\Lambda_2}{\Lambda_1}-b\right\}+i\left[\left\{\left(\frac{n\pi}{\Lambda_1}-2a\right)\frac{\Lambda_2}{\Lambda_1}-b\right\}a+b\left(\frac{n\pi}{\Lambda_1}-a\right)\right],\]
  we must enforce the conditions
  \[b=\frac{\Lambda_2}{\Lambda_1}a,\qquad a\neq \frac{n\pi}{2\Lambda_1}.\]
  So the first wave vectors can be equivalently written as
  \[k_1=a+i\frac{\Lambda_2}{\Lambda_1}a,\qquad h_2=\left(\frac{n\pi}{\Lambda_1}-a\right)+i\left\{\left(\frac{n\pi}{\Lambda_1}-3a\right)\frac{\Lambda_2}{\Lambda_1}\right\}.\]
  This means the metric component is simply a polynomial in the wave vector parameter $a$:
  \[p^2=\frac{\Lambda_2}{\Lambda_1}\left(1-\left(\frac{\Lambda_2}{\Lambda_1}\right)^2\right)a+\left(3\left(\frac{\Lambda_2}{\Lambda_1}\right)^2-1\right)a^2.\]

\section{The Metric Under a Dehn Twist}
A $(p,q)$ Dehn twist is a surgery on manifolds in which the torus is removed, twisted by
\[(\Lambda_1,\Lambda_2)\to(p\Lambda_1,q\Lambda_2),\]
and glued back in. Interestingly, the metric is not invariant under this surgery:
\[p'^2=\frac{n\pi}{p\Lambda_1}\left(1-\left(\frac{q\Lambda_2}{p\Lambda_1}\right)^2\right)a+\left(3\left(\frac{q\Lambda_2}{p\Lambda_1}\right)^2-1\right)a^2.\]
In general, the metric will only be invariant in the trivial $p=q=1$ case, or for $p=q$ when the wave vector parameter $n=p$. Other cases will depend on the choice of $\Lambda_1$ and $\Lambda_2$.

\section{Summary}
We have presented a new solution to the generalized Weierstrass-Enneper system, in which we've explicitly specified the wave functions and demonstrated that they describe a torus in $\mathbb{R}^4$. Besides adding to the body of knowledge of this system, which has found application in a variety of mathematical subfields, we have shown that under surgery, the metric of our presentation is generically modified. This is significant, since this surgery typically takes place at the level of topology - for example, $(2,1)$ surgery on $\mathbb{S}^3$ results in $\mathbb{RP}^3$\cite{Gompf-Stipsicz-1999}. However, since any closed 3-manifold can be realized in this way \cite{Rohlin-1951}, we know the geometry is indeed changed under it. It is our hope that by using the Weierstrass-Enneper representation to make the change in the geometry explicit (at the level of the metric), this representation can be used to study other systems in which the Dehn surgery plays a key role. A illustrative example is found in the study of exotic smooth structure\cite{Asselmeyer-Maluga-Brans-2014} (or in a similar way with Fintushel-Stern knot surgery\cite{Asselmeyer-Maluga-Brans-2015}), where this technique could be used to study the physical properties of manifolds which are abstractly generated via this surgery.





\bibliography{WeierstrassTorusBib}

\end{document}